\PassOptionsToPackage{usenames, dvipsnames}{xcolor}
\documentclass{toc}
\usepackage{enumerate}
\usepackage{listings}
\usepackage{amsmath,amsthm}
\usepackage{stmaryrd}
\usepackage{microtype}
\usepackage{hyperref,doi}
\usepackage{graphicx}
\usepackage{tikz}
\usepackage{caption,subcaption}
\usepackage{listings}
\usepackage[numbers]{natbib}
\usepackage{courier}
\theoremstyle{plain}
\newtheorem{observation}[theorem]{Observation}

\hypersetup{
  colorlinks=true,
  breaklinks,
  anchorcolor=[rgb]{0.2, 0.6, 0.2}, %
  citecolor=[rgb]{0.2, 0.6, 0.2}, %
  urlcolor=[rgb]{0.2, 0.2, 0.6}, %
  citecolor=[rgb]{0.2, 0.6, 0.2} %
}

\usepackage{thmtools,thm-restate}
\usepackage[linesnumbered,vlined,boxed,ruled]{algorithm2e}
\makeatletter
\newcommand{\RemoveAlgoNumber}{\renewcommand{\fnum@algocf}{\AlCapSty{\AlCapFnt\algorithmcfname}}}
\newcommand{\RevertAlgoNumber}{\algocf@resetfnum}
\makeatother

\DeclareMathOperator{\poly}{poly}

\newcommand{\floor}[1]{\left\lfloor #1 \right\rfloor}
\newcommand{\ceil}[1]{\left\lceil #1 \right\rceil}
\newcommand{\set}[1]{\left\{ #1 \right\}}
\newcommand{\abs}[1]{\left| #1 \right|}
\newcommand{\inparen}[1]{\left( #1 \right)}
\newcommand{\inbrace}[1]{\left\{ #1 \right\}}
\newcommand{\pfrac}[2]{\inparen{\frac{1}{2}}}
\newcommand{\setdef}[2]{\inbrace{ #1 \ : \ #2}}

\newcommand{\NWify}[2]{#1\llbracket #2 \rrbracket_{\mathrm{NW}}}
\newcommand{\ilog}[1]{\log^{\circ #1}}
\newcommand{\F}{\mathbb{F}}

\newcommand{\VNP}{\mathsf{VNP}}

\newcommand{\veca}{\mathbf{a}}
\newcommand{\vecx}{\mathbf{x}}
\newcommand{\vecy}{\mathbf{y}}
\newcommand{\vecz}{\mathbf{z}}

\usepackage{nth}
\usepackage{intcalc}
\usepackage{etoolbox}
\usepackage{xstring}

\usepackage{ifpdf}
\ifpdf
\else
\usepackage[quadpoints=false]{hypdvips}
\fi

\newcommand{\ECCCurl}[2]{http://eccc.hpi-web.de/report/\ifnumcomp{#1}{>}{93}{19}{20}#1/#2/}

\newcommand{\shortECCC}[2]{\texttt{\href{http://eccc.hpi-web.de/report/\ifnumcomp{#1}{>}{93}{19}{20}#1/#2/}{eccc:TR#1-#2}}}

\newcommand{\parseECCC}[1]{%
\StrSubstitute{#1}{TR}{}[\tmpstring]%
\IfSubStr{\tmpstring}{/}{ %
\StrBefore{\tmpstring}{/}[\ecccyear]%
\StrBehind{\tmpstring}{/}[\ecccreport]%
}{%
\StrBefore{\tmpstring}{-}[\ecccyear]%
\StrBehind{\tmpstring}{-}[\ecccreport]%
}%
\shortECCC{\ecccyear}{\ecccreport}}

\tocdetails{
  title = {Near-Optimal Bootstrapping of Hitting Sets for Algebraic Models},
  author = {Mrinal Kumar, Ramprasad Saptharishi, and Anamay Tengse},
  plaintextauthor = {Mrinal Kumar, Ramprasad Saptharishi, Anamay Tengse},
  acmclassification = {F.2.1, F.1.3},
  amsclassification = {68W30, 68Q87, 68Q06},
  keywords = {algebraic circuits, polynomial identity testing, derandomization, hardness-randomness, bootstrapping}
}

\tocdetails{%
volume = 19,
number = 12,
year = 2023,
firstpage = 1,
received = {April 16, 2019},
revised = {August 5, 2021},
published = {December 31, 2023},
doi = {10.4086/toc.2023.v019a012},
}   %

\begin{document}

\begin{frontmatter}

\title{Near-Optimal Bootstrapping of Hitting Sets\\for Algebraic Models\titlefootnote{A preliminary version of this paper appeared in the 
\href{https://doi.org/10.1137/1.9781611975482.40}{%
Proceedings of the 30th Annual ACM-SIAM Symposium on Discrete Algorithms
(SODA 2019)} \cite{conf-version}.}}

\author[mrinal]{Mrinal Kumar\thanks{Tata Institute of Fundamental Research, Mumbai, India. A part of this work was done during postdoctoral stays at Simons Institute for the Theory of Computing, Berkeley, USA, Center for Mathematical Sciences and Applications at Harvard, and while visiting TIFR, Mumbai.}}
\author[ramprasad]{Ramprasad Saptharishi\thanks{Tata Institute of Fundamental Research, Mumbai, India. Research supported by the Department of Atomic Energy, Government of India, under project 12-R\&D-TFR-5.01-0500}}
\author[anamay]{Anamay Tengse\thanks{Reichman University, Herzliya. A major part of this work was done as a student at TIFR, Mumbai, when the author was supported by a fellowship of the DAE, Govt. of India.}}

\begin{abstract}
The %
Polynomial Identity Lemma  %
(also called the ``Schwartz--Zippel lemma'') states that any nonzero polynomial $f(x_1,\ldots, x_n)$ of degree at most $s$ will evaluate to a nonzero value at some point on %
any  %
grid $S^n \subseteq \F^n$ with $\abs{S} > s$.
  Thus, there is an explicit \emph{hitting set} for all $n$-variate degree-$s$, size-$s$ algebraic circuits  of size $(s+1)^n$.

  In this paper, we prove the following results:
  \begin{itemize}
    
  \item Let $\epsilon > 0$ be a constant.
  For a sufficiently large constant $n$, and all $s > n$, if we have an explicit hitting set of size $(s+1)^{n-\epsilon}$ for the class of $n$-variate degree-$s$ polynomials that are computable by algebraic circuits of size~$s$, then for all large $s$, we have an explicit hitting set of size $s^{\exp(\exp (O(\log^\ast s)))}$ for $s$-variate circuits of degree~$s$ and size~$s$.

That is, if we can obtain a barely non-trivial exponent 
(a factor-$s^{\Omega(1)} $  %
improvement) compared to the trivial %
$(s+1)^{n}$-size  %
hitting set even for
constant-variate %
circuits, we can get an almost complete derandomization of PIT.
    
    \item The above result holds when ``circuits'' are replaced by ``formulas'' or ``algebraic branching programs.''  
  \end{itemize}

  This extends a recent surprising result of Agrawal, Ghosh and Saxena~(STOC 2018, PNAS 2019) who proved the same conclusion for the class of algebraic circuits, if the hypothesis provided a hitting set of size at most $\inparen{s^{n^{0.5 - \delta}}}$ (where $\delta>0$ is any constant).
  Hence, our work significantly weakens the hypothesis of Agrawal, Ghosh and Saxena to  only require a slightly non-trivial saving over the trivial hitting set, and also presents  the first such result for algebraic formulas.
\end{abstract}

\end{frontmatter}

\section{Introduction}

Multivariate polynomials are the primary protagonists in the field of algebraic complexity and algebraic circuits form a natural robust model of computation for multivariate polynomials.
For completeness, we now define algebraic circuits : an algebraic circuit is  a directed acyclic graph with internal gates labeled by $+$ (addition) and $\times$ (multiplication), and with leaves labeled by either variables or field constants; computation flows in the natural way.

In the field of algebraic complexity, much of the focus has been restricted to studying $n$-variate polynomials whose degree is bounded by a polynomial function in $n$, and such polynomials are called \emph{low-degree polynomials}.
This restriction has several \emph{a priori} and \emph{a posteriori} motivations, and excellent discussions of this can be seen in the thesis of Forbes~\cite[Section 3.2]{F14} and Grochow's answer~\cite{G13} on \texttt{cstheory.SE}.
The central question in algebraic complexity is to find a family of low-degree polynomials that requires large algebraic circuits to compute it.
Despite having made substantial progress in various subclasses of algebraic circuits %
(see, \eg, the  %
surveys \cite{SY10,S15}), the current best lower bound for general algebraic circuits is merely an $\Omega(n \log d)$ lower bound of Baur and Strassen~\cite{BS83}.

An interesting approach towards proving lower bounds for algebraic circuits is via showing good \emph{upper bounds} for the \emph{algorithmic} task of \emph{polynomial identity testing.}
Our results in this paper deal with this problem, and we elaborate on this now.

\subsection{Polynomial Identity Testing}

Polynomial identity testing (PIT\footnote{We use the abbreviation PIT for both the noun `polynomial identity test' and gerund/adjective `polynomial identity testing'. The case would be clear from context.}) is the algorithmic task of checking if a given algebraic circuit $C$ of size~$s$ computes the identically zero polynomial.
As discussed earlier, although a circuit of size~$s$ can compute a polynomial of degree $2^s$, this question typically deals only with circuits whose \emph{formal degree}\footnote{This is defined inductively by setting the formal degree of leaves as $1$, and taking the sum at every multiplication gate and the max at every sum gate.}
is bounded by the size of the circuit.

PIT is an important algorithmic question
in  %
its own right, and many classical results such as the primality testing algorithm~\cite{AKS04}, $\mathsf{IP} = \mathsf{PSPACE}$~\cite{LFKN90,S90a}, algorithms for graph matching~\cite{MVV87,FGT16,ST17} all have
polynomial identity tests at their %
core.

This algorithmic question has two flavors: whitebox PIT and blackbox PIT.
Whitebox polynomial identity tests consist of algorithms that can inspect the circuit (that is, look at the underlying gate connections etc.)
to decide whether the circuit computes the zero polynomial or not.
A stronger algorithm is a \emph{blackbox polynomial identity test} where the algorithm is only provided basic parameters of the circuit (such as its size, the number of variables, a bound on the formal degree) and only has evaluation access to the circuit $C$.
Hence, a blackbox polynomial identity test for a class $\mathcal{C}$ of circuits is essentially just a list of evaluation points $H \subseteq \F^n$ such that every nonzero circuit $C\in \mathcal{C}$ is guaranteed to have some $\veca\in H$ such that $C(\veca) \neq 0$. Such sets of points are also called \emph{hitting sets} for $\mathcal{C}$. Therefore, the running time of a blackbox PIT algorithm is  given by the size of the \emph{hitting set}, the time taken to generate it given the parameters of the circuit, and the time taken to evaluate the circuit on these points.
We shall say that a hitting set $H$ is \emph{explicit} if there is an algorithm that, given the parameters $n,d,s$, outputs the set $H$ in time that is polynomial in the \emph{bit complexity}\footnote{Throughout the paper, we will only talk about the sizes of the hitting sets as their bit complexities %
(\ie, the total bit-lengths) %
are always bounded by a polynomial in their respective sizes. %
We provide more details about this analysis in \expref{subsection}{subsec:algo}.} of $H$.

The classical Polynomial Identity Lemma\footnote{A
discussion %
of  %
the history of this result can be found %
in~\cite{BCPS18}.%
}~\cite{O22,DL78,Z79,S80} states the following.
\begin{lemma}[Polynomial Identity Lemma]   \label{lem:PIL}
Any nonzero polynomial $f(x_1,\ldots, x_n)$ of degree at most $d$ will 
evaluate to a nonzero value at a randomly chosen point from a
finite  %
grid $S^n \subseteq \F^n$ with probability at least $1 - \frac{d}{|S|}$.
\end{lemma}
This   %
automatically yields a \emph{randomized} polynomial-time blackbox PIT algorithm, and also an explicit hitting set of size $(d+1)^n$, for the class of $n$-variate formal degree-$d$ polynomials.
Furthermore, a simple counting/dimension argument also says that there exist (non-explicit) %
$\poly(s)$-size   %
hitting sets for the class of polynomials computed by size-$s$ algebraic circuits.
The major open question is to find a better \emph{deterministic} algorithm for this problem, and the task of constructing deterministic PIT algorithms is intimately connected with the question of proving explicit lower bounds for algebraic circuits.

Heintz and Schnorr~\cite{HS80}, and Agrawal~\cite{A05a} observed that given an explicit hitting set for size-$s$ circuits, any nonzero polynomial that is designed to vanish on every point of the hitting set cannot be computable by size-$s$ circuits.
By tailoring the number of variables and degree of the polynomial in this observation, they showed that polynomial-time blackbox PITs yield an $\mathsf{E}$-computable family $\set{f_n}$ of $n$-variate multilinear polynomials that require 
$2^{\Omega(n)}$-size  %
circuits.
This connection between PIT and lower bounds was strengthened further by Kabanets and Impagliazzo~\cite{KI04} who showed that explicit families of hard functions can be used to give non-trivial derandomizations for PIT.
Thus, the question of proving explicit lower bounds and the task of finding upper bounds for PIT are essentially two sides of the same coin.

\subsection{Bootstrapping}

A recent result of Agrawal, Ghosh and Saxena~\cite{AGS19} showed, among other things, the following surprising result: blackbox PIT algorithms for size-$s$ and $n$-variate circuits with running time as bad as $\inparen{s^{n^{0.5 - \delta}}}$, where $\delta > 0$ is a constant, can be used to construct blackbox PIT algorithms for size-$s$ circuits with running time $s^{\exp(\exp (O(\log ^\ast s)))}$.
Note that $\log^\ast n$ refers to the smallest $i$ such that the $i$-th iterated logarithm $\log^{\circ i}(n)$ is at most $1$.
This shows that good-enough derandomizations of PIT would be sufficient to get a nearly complete derandomization.
Their proof uses a novel \emph{bootstrapping} technique where they use the connection between hardness and derandomization repeatedly so that by starting with a weak hitting set we can obtain better and better hitting sets.

One of the open questions of Agrawal, Ghosh and Saxena~\cite{AGS19} was whether the hypothesis can be strengthened to a \emph{barely} non-trivial derandomization.
That is, suppose we have a blackbox PIT algorithm, for the class of
$n$-variate, size-$s$ circuits %
that runs in time $s^{o(n)}$, can we use this to get a nearly complete derandomization?
Note that we have a trivial $(s+1)^{n} \cdot \poly(s)$ algorithm from the 
Polynomial Identity Lemma (\expref{Lemma}{lem:PIL}).  %
Our main result is an affirmative answer to this question in a very strong sense. Furthermore, our result holds for \emph{typical} subclasses that are reasonably well-behaved under composition.
Formally, we prove the following theorem.

\begin{restatable}[Bootstrapping PIT for algebraic formulas, branching programs and circuits]{theorem}{maintheorem}
\label{thm:mainTheorem-2} Let $\epsilon > 0$ and $n \geq 2$ be constants. Suppose that, for all large enough $s$, there is an explicit hitting set of size
$s^{n-\epsilon}$ for all degree-$s$, size-$s$ algebraic formulas 
(or algebraic branching programs or circuits) %
over $n$ variables.
Then, for all large $s$, there is an explicit hitting set of size $s^{\exp(\exp(O(\log^{*}s)))}$ for the class of degree-$s$, size-$s$ algebraic formulas
(or algebraic branching programs or circuits, respectively)  %
over $s$ variables.
\end{restatable}

\begin{remark}
While the hypothesis of \expref{Theorem}{thm:mainTheorem-2} above assumes that $n$ is a constant, qualitatively similar results continue to hold even when $n$ is slowly growing with $s$. For simplicity of notation, we work with constant $n$ throughout the paper and discuss the extension to growing $n$ in more detail in \expref{subsection}{sec:growing n}. 
\end{remark}

Note that $(s+1)^{n-\epsilon} = s^{n-\epsilon} \cdot \inparen{1 + \frac{1}{s}}^{n-\epsilon} < \eee \cdot s^{n-\epsilon} < s^{n-\epsilon'}$ for some other constant $\epsilon' > 0$ since $s$ is large enough. Hence, for this theorem, there is no qualitative difference if the hitting set had size $(s+1)^{n-\epsilon}$ instead of $s^{n-\epsilon}$.
We also note that as far as we understand, such a statement for classes such as algebraic formulas, even with the stronger hypothesis of there being an $s^{O(n^{(1/2)-\epsilon})}$, did not follow from the results in \cite{AGS19}.
We elaborate more on this, and the differences between our proof and theirs in the next subsection.

An interesting, albeit simple corollary of the above result is the following statement.

\begin{corollary}[From slightly non-trivial PIT to lower bounds] \label{cor:lb}
Let $\epsilon > 0$ and $k \geq 2$ be constants.
Suppose that, for all large enough $s$, there is an explicit hitting set of size $\inparen{s^{k-\epsilon}}$ for all degree-$s$, size-$s$ algebraic formulas 
(or algebraic branching programs or circuits)  %
over $k$ variables.
Then, for every function $d:\N \rightarrow \N$, there is a polynomial family $\{f_{n}\}$, where $f_n$ is %
an %
$n$-variate  %
of  %
degree $d(n)$, and for every large enough $n$, $f_n$ cannot be computed by algebraic formulas
or algebraic branching programs or circuits %
of size smaller than $\binom{n+d}{d}^{\frac{1}{\exp(\exp(O(\log^{*}nd)))}}$.
Moreover, there is an algorithm which when given as input an
$n$-variate  %
monomial of degree~$d$, outputs its coefficient in $f_n$ in deterministic time $\binom{n+d}{d}$.
\end{corollary}

Thus, a slightly non-trivial blackbox PIT algorithm leads to hard families with near optimal hardness (as any $n$-variate polynomial of total degree~$d$ trivially has a circuit of size $\binom{n+d}{d}$).
In a recent result concerning \emph{non-commutative} algebraic circuits, Carmosino, Impagliazzo, Lovett and Mihajlin~\cite{CILM18} showed that given an explicit polynomial family of constant degree which requires
superlinear  %
size  %
non-commutative circuits, one can obtain explicit polynomial families of exponential hardness.
Besides the obvious differences in the statements, one important point to note is that the notions of explicitness in the conclusions of the two statements are different from each other.
In~\cite{CILM18}, the final exponentially hard polynomial family is in $\VNP$ provided the initial polynomial family is also in $\VNP$.
On the other hand, for our result, we can say that the hard polynomial family obtained in the conclusion is explicit in the sense that its coefficients are computable in deterministic time $\binom{n+d}{d}$.
Another difference between~\expref{Corollary}{cor:lb} and the main result of~\cite{CILM18} is in the hypothesis.
From a non-trivial hitting set, we can obtain a large class of lower bounds by varying parameters appropriately (see \expref{Theorem}{thm:HS}), however the main result of \cite{CILM18} starts with a lower bound for a single family. In that regard, our hypothesis appears to be much stronger and slightly non-standard.
We discuss this issue in some detail at the end of the next section.

In another relevant result, Jansen and Santhanam \cite{JS12} showed that marginal improvements to known hitting set constructions imply lower bounds for the permanent polynomial. In particular, they show that a ``sufficiently succinct'' hitting set of size~$d$, for univariates of degree~$d$ that have \emph{constant-free} algebraic circuits of small size, would imply that the permanent polynomial requires 
superpolynomial-size  %
\emph{constant-free} algebraic circuits. Note that even though their hypothesis needs a much weaker improvement in the size of the hitting set when compared to ours, the hitting set is additionally required to be ``succinct''\footnote{They require their hitting sets to be encoded by uniform $\mathsf{TC}^0$ circuits of appropriately small size. See~\cite{JS12} for details.}, which makes it difficult to compare the two hypotheses.

\subsection{Proof overview}

The basic intuition for the proofs in this paper, and as per our understanding also for the proofs of the results in the work of Agrawal, Ghosh and Saxena~\cite{AGS19}, comes from the results of Kabanets and Impagliazzo~\cite{KI04}, and those of Heintz and Schnorr~\cite{HS80} and Agrawal~\cite{A05a}.
We start by informally stating these results.

\begin{theorem}[Informal, Heintz and Schnorr~\cite{HS80},
  Agrawal~\cite{A05a}]\label{thm:HS} Let $H(n, d, s)$ be an explicit hitting set for circuits of size~$s$, degree~$d$ in $n$ variables.
  Then, for every $k\leq n$ and $d'$ such that $d'k \leq d$ and $(d'+1)^{k} > \abs{H(n, d, s)}$, there is a nonzero polynomial in $k$ variables and individual degree~$d'$ that vanishes on the hitting set $H(n, d, s)$, and hence cannot be computed by a circuit of size~$s$.
\end{theorem}
In a nutshell, given an explicit hitting set, we can obtain hard polynomials.
In fact, playing around with the parameters $d'$ and $k\leq n$, we can get a hard polynomial in $k$ variables,  %
of   %
degree $kd'$ for all $k, d'$ satisfying $d'k < d$ and $(d'+1)^{k} > \abs{H(n, d, s)}$.

We now state a result of Kabanets and Impagliazzo~\cite{KI04} that shows that hardness can lead to derandomization.
\begin{theorem}[Informal, Kabanets and Impagliazzo~\cite{KI04}]\label{thm:KI}
  A superpolynomial lower bound for algebraic circuits for an explicit family of polynomials implies a deterministic blackbox PIT algorithm for all algebraic circuits in $n$ variables and degree~$d$ of size $\poly(n)$ that runs in time $d^{O\inparen{n^{\epsilon}}}$ for every $\epsilon > 0$.
\end{theorem}
Now, we move on to the main ideas in our proof.
Suppose we have non-trivial hitting sets for size-$s$, degree $d \leq s$ circuits on $n$ variables.
The goal is to obtain a blackbox PIT for circuits of size~$s$, degree~$s$ on $s$ variables with a much better dependence on the number of variables.

Observe that if the number of variables was much much smaller than $s$, say at most a constant, then the hitting set in the hypothesis has a polynomial dependence on $s$, and we are done.
We will proceed by presenting \emph{variable reductions} to eventually reach this stage.
With this in mind, the hitting sets for
$s$-variate %
circuits in the conclusion of~\expref{Theorem}{thm:mainTheorem-2} are designed iteratively starting from hitting sets for circuits with very few variables.
In each iteration, we start with a hitting set for size-$s$, degree $d \leq s$ circuits on $n$ variables with some dependence on $n$ and obtain a hitting set for size-$s$, degree $d \leq s$ circuits on $m = 2^{n^{\delta}}$ variables (for some $\delta > 0$), that has a \emph{much} better dependence on $m$.
Then, we repeat this process till the number of variables increases up to $s$, which takes $O(\log^\ast s)$ iterations.
We now briefly outline the steps in each such iteration.
\begin{itemize}

  \item {\bf Obtaining a family of hard polynomials : } The first step is to obtain a family of explicit hard polynomials from the given hitting sets.
  This step is done via~\expref{Theorem}{thm:HS}, which simply uses interpolation to find a nonzero polynomial $Q$ in $k$ variables and degree~$d$ that vanishes on the hitting set for size-$s'$, degree-$d'$ circuits on $n$ variables, for some $s', d'$ to be chosen appropriately.

  \item {\bf Variable reduction using $Q$ : } Next we take a combinatorial design (see \expref{Definition}{defn:NWdesign}) $\{S_1, S_2, \ldots, S_m\}$, where each $S_i$ is a subset of size~$k$ of a universe of size $\ell=\poly(k)$, and $\abs{S_i \cap S_j} \ll k$.
  Consider the map $\Gamma : \F[x_1, x_2, \ldots, x_m] \rightarrow \F[y_1, y_2, \ldots, y_{\ell}]$ given by the substitution $\Gamma(C(x_1, x_2, \ldots, x_m))$ $= C\inparen{Q(\vecy\mid_{S_1}),Q(\vecy\mid_{S_2}),\ldots, Q(\vecy\mid_{S_m})}$.
  As Kabanets and Impagliazzo show in the proof of~\expref{Theorem}{thm:KI}, $\Gamma$ preserves the nonzeroness of all algebraic circuits of size~$s$ on $m$ variables, provided $Q$ is hard enough.

  We remark that our final argument for this part is slightly simpler than that of Kabanets and Impagliazzo, and hence our results also hold for algebraic formulas.
  In particular, we do not need Kaltofen's seminal result that algebraic circuits are closed under polynomial factorization, whereas the proof of Kabanets and Impagliazzo crucially uses Kaltofen's result~\cite{K89}.
  This comes from the simple, yet effective, observation that if $Q$ vanishes on some hitting set, then so does any multiple of $Q$.
  This allows us to use the hardness of \emph{low-degree} multiples of $Q$, and so, we do not need any complexity guarantees on factors of polynomials.

  It is worth mentioning that the result of Kabanets and Impagliazzo~\cite{KI04} is the algebraic analogue of the famous result of Nisan and Wigderson~\cite{NW94} from the boolean world.
  However, the algebraic setting allows us to construct polynomials with hardness that is \emph{super-exponential} in the number of variables.
  This is a key reason for bootstrapping to work, and we shall elaborate a bit more on this later in this section.

  \item {\bf Blackbox PIT for $m$-variate circuits of size~$s$ and degree~$s$ : } We now take the hitting set given by the hypothesis for the circuit $\Gamma(C)$ (invoked with appropriate size and degree parameters) and evaluate $\Gamma(C)$ on this set.
  From the discussion so far, we know that if $C$ is nonzero, then $\Gamma(C)$ cannot be identically zero, and hence it must evaluate to a nonzero value at some point on this set.
  The number of variables in $\Gamma(C)$ is at most $\ell = \poly\log {m}$, whereas its size turns out to be \emph{not too much} larger than $s$.
  Hence, the size of the hitting set for $C$ obtained via this argument turns out to have a better dependence on the number of variables $m$ than the hitting set in the hypothesis.

\end{itemize}

To prove~\expref{Corollary}{cor:lb}, we let $t(n) = \exp(\exp(O(\log^{*}n)))$.
Now, we invoke the the conclusion of~\expref{Theorem}{thm:mainTheorem-2} with $s = {\binom{n+d}{d}}^{\frac{1}{10t(n)}}$.
Thus, we get an explicit hitting set $H$ of size ${\binom{n+d}{d}}^{\frac{1}{10}}$ for $n$-variate  %
circuits of size~$s$ and degree~$d$.
We now use~\expref{Theorem}{thm:HS} to get a nonzero polynomial of degree~$d$ and $n$ which vanishes on the set $H$ and hence cannot be computed by circuits of size at most $s$.
We skip the rest of the details.

\paragraph*{Why does bootstrapping work?}
As far as we understand, the primary reason that makes such bootstrapping results feasible is the following observation from the results of Heintz and Schnorr, and Agrawal~\cite{HS80,A05a}: Given a single hitting set, we can obtain a \emph{family} of lower bounds by varying the degree and the number of variables in the interpolating polynomial. 
It turns out that in the result of Kabanets and Impagliazzo~\cite{KI04} that converts a hard polynomial $P$ into a hitting set, the proof of this conversion has different sensitivities to the degree of $P$ and the number of variables it depends on.
The combination of both these facts allows us to start with a moderately non-trivial hitting set, obtaining a hard polynomial from it of the \emph{right} degree and number of variables, and use that to obtain a hitting set which is significantly better than what we started with.
In particular, by choosing a degree that is polynomial in the hitting set size, say $\abs{H}^{0.1} $, we can obtain a \emph{constant-variate} polynomial that vanishes on $H$, for any $H$.
This is impossible in the boolean world where the best hardness one could hope for is exponential in the number of variables.
This, in our opinion, is a high level picture of why bootstrapping works in the algebraic world.\\

\paragraph*{Similarities and differences with the proof of Agrawal, Ghosh and Saxena~\cite{AGS19}. }
The high level outline of our proof is essentially the same as that in \cite{AGS19}.
However, there are some differences that make our final arguments shorter, simpler and more robust than those of Agrawal, Ghosh and Saxena thus leading to a stronger and near optimal bootstrapping statement in~\expref{Theorem}{thm:mainTheorem-2}.
Moreover, as we already alluded to, our proof extends to formulas and algebraic branching programs as well, whereas, to the best of our understanding, the proofs in \cite{AGS19} do not.
We now elaborate on the differences.

One of the main differences between the proofs in this paper and those in \cite{AGS19} is in the use of the result of Kabanets and Impagliazzo~\cite{KI04}. %
Agrawal, Ghosh and Saxena use this result as a blackbox to get deterministic PIT using hard polynomials.
The result of Kabanets and Impagliazzo~\cite{KI04} requires a result of Kaltofen, which shows that low-degree %
algebraic circuits are closed under polynomial factorization.
That is, if a degree-$d$, $n$-variate polynomial $P$ has a circuit of size at most $s$, then any factor of $P$ has a circuit of size at most $(snd)^e$ for a constant $e$.
Such a closure result is not known to be true for formulas\footnote{Even for algebraic branching programs such a result was shown only recently (after our work) by Sinhababu and Thierauf~\cite{ST20}.}, and hence the results in \cite{AGS19} do not seem to extend to these settings.
Also, the removal of any dependence on the ``factorization exponent'' $e$ is a key ingredient in our proof as it allows us to start with a hypothesis of a barely non-trivial hitting set.
To see this more clearly, suppose we wish to start with a hypothesis that gives hitting sets of size $s^{g(n)} $ for $n$-variate circuits of size and degree~$s$.
It is then not too difficult to infer from \expref{Lemma}{lem:main-bootstrap} that for any proof of ``variable reduction'' using the \emph{hybrid argument} of Kabanets and Impagliazzo~\cite{KI04}, we require $g(n)$ to satisfy the relation $ e \cdot g(n) \leq n$.
This would mean that $s^{g(n)}$ can never be something like $s^{n - \epsilon}$ \ie, a $\poly(s)$ factor better than the trivial $s^n$.

The other main difference between our proof and that in \cite{AGS19} is rather technical but we try to briefly describe it.
This is in the choice of combinatorial designs.
The designs used in this paper are based on the standard Reed-Solomon code and they yield larger set families than the designs used by \cite{AGS19}.
However, even without these improved design parameters, our proof can be used to provide the same conclusion when starting off with a hitting set of size $s^{n^{\delta}} $, instead of the hypothesis of~\expref{Theorem}{thm:mainTheorem-2}\footnote{Even though $n$, $\epsilon$ and $\delta$ are constants, $n^{\delta} $ and $(n-\epsilon)$ are qualitatively different, as $\epsilon$ is independent of $n$.}.

Also, their proof is quite involved and we are unsure if there are other constraints in their proof that force such choices of parameters.
Our proof, though along almost exactly the same lines, appears to be more transparent and more malleable with respect to the choice of parameters.

\paragraph*{The strength of the hypothesis. } The hypothesis of~\expref{Theorem}{thm:mainTheorem-2} and also those of the results in the work of Agrawal, Ghosh and Saxena~\cite{AGS19} is that we have a non-trivial explicit hitting set for algebraic circuits of size~$s$, degree~$d$ on $n$ variables where $d$ and $s$ could be arbitrarily large as functions of $n$. This seems like an extremely strong assumption, and also slightly non-standard in the following sense. In a typical setting in algebraic complexity, we are interested in PIT for size-$s$, degree-$d$ circuits on $n$ variables where $d$ and  $s$ are polynomially bounded in the number of variables $n$. A natural open problem here, which would be a more satisfying statement to have, would be to show that one can weaken the hypothesis in~\expref{Theorem}{thm:mainTheorem-2} to only hold for circuits whose degree and size are both polynomially bounded in $n$. It is not clear to us if such a result can be obtained using the current proof techniques, or is even true.

Having noted that our hypothesis is very strong, and perhaps even slightly unnatural with respect to the usual choice of parameters in the algebraic setting, we remark that our hypothesis does in fact follow from the assumptions that the Permanent is hard for Boolean circuits, and the Generalized Riemann Hypothesis (GRH).
The proof is essentially the same as that of Corollary 1 in the work of Jansen and Santhanam~\cite{JS12}.
The only difference is that while Jansen and Santhanam show that there are non-trivial explicit\footnote{In fact their notion of explicitness is stronger than ours.} hitting sets for univariate polynomials with small circuits assuming the hardness of Permanent for Boolean circuits and the GRH, here we have to work with circuits computing multivariate polynomials. %
At a high level, the proof in~\cite{JS12} proceeds by constructing a pseudorandom generator for Boolean circuits of appropriate size assuming the hardness of permanent for Boolean circuits.
Then, the set of binary strings in the output of this generator is interpreted in a natural way as an integer.
This gives us a small set of integer points, which can be constructed deterministically.
Then they argue that there is no constant free algebraic circuit of small size which vanishes on all these integer points.
The proof of this step is %
by  %
contradiction, where they assume the existence of such a constant free algebraic circuit to construct a Boolean circuit of small size which is not fooled by the aforementioned Boolean pseudorandom generator.
For algebraic circuits which are not constant free and are allowed to use arbitrary field constants and hence cannot be efficiently simulated by a Boolean circuit, they assume the GRH to reduce to the case of constant free circuits in a fairly standard way.
For our setting, we interpret the output of the Boolean pseudorandom generator as not just a single integer point, but a $k$ tuple of integers points.
These set of points in $\Z^k$ form our candidate hitting set.
The rest of the proof carries over without any changes. We refer the interested reader to~\cite{JS12} for further details.

\begin{remark}
Throughout the paper, we shall assume that there are suitable $\floor{\cdot}$'s or $\ceil{\cdot}$'s if necessary so that certain parameters chosen are integers. We avoid writing this purely for the sake of readability.

All results in this paper continue to hold for the underlying model of algebraic formulas, algebraic branching programs or algebraic circuits. In fact, the results also extend to the model of \emph{border} of algebraic formulas, algebraic branching programs or algebraic circuits \ie if there is a slightly non-trivial hitting set for polynomials in the border of these classes, then our main theorem gives a highly non-trivial explicit hitting set for these polynomials. Since our proofs extend as it is to this setting with essentially no changes, we skip the details for this part, and confine our discussions in the rest of the paper to just standard algebraic formulas.
\end{remark}

\section{Preliminaries}

\subsection*{Notation}

\begin{itemize}\itemsep 0pt
\item For a positive integer $n$, we use $[n]$ to denote the set $\set{1,2,\ldots, n}$.
\item We use boldface letters such as $\vecx_{[n]}$ to denote a set $\set{x_1,\ldots, x_n}$. We drop the subscript whenever the number of elements is clear or irrelevant in the context.
\item For a polynomial $f(x_1,\ldots, x_n)$, we shall say its \emph{individual degree} is at most $k$ to mean that the exponent of any of the $x_i$'s in any monomial is at most $k$.
\end{itemize}

We now define some standard notions we work with, and state some of the known results that we use in this paper.

\subsection{Algebraic models of computation}

Throughout the paper we would be dealing with some standard algebraic models and we define them formally for completeness.

\begin{definition}[Algebraic branching programs (ABPs)]~\label{def:abp}
  An \emph{algebraic branching program} in variables $\{x_1, x_2, \ldots, x_n\}$ over a field $\F$ is a directed acyclic graph with a designated \emph{starting vertex} $s$ with in-degree zero, a designated \emph{end vertex} $t$ with out-degree zero, and the edge between any two vertices labeled by an affine form from $\F[x_1, x_2, \ldots, x_n]$.
  The polynomial computed by the ABP is the sum of all weighted paths from $s$ to $t$, where the weight of a directed path in an ABP is the product of labels of the edges in the path.

  The size of an ABP is defined as the number of edges in the underlying graph.
\end{definition}
\begin{definition}[Algebraic formulas]
  An algebraic circuit is said to be a \emph{formula} if the underlying graph is a tree.
  The size of a formula is defined as the number of leaves.

  The notation $\mathcal{C}(n,d,s)$ will be used to denote the class of $n$-variate\footnote{This class may also include polynomials that actually depend on fewer variables but are masquerading %
as  %
$n$-variate polynomials.}
 polynomials of degree at most $d$ that are computable by formulas of size at most $s$.
\end{definition}

We will use the following folklore algorithm for computing univariate polynomials, often attributed to Horner\footnote{Though this method was discovered at least 800 years earlier by the Iranian mathematician and astronomer Sharaf al-D\={\i}n \d{T}\={u}s\={\i} %
(see  %
Hogendijk~\cite{H89}). }. We also include a proof for completeness.
\begin{proposition}[Horner rule]\label{prop:horner}
Let $P(x) = \sum_{i = 0}^d p_ix^i$ be a univariate polynomial of degree  $d$ over any field $\F$. Then, $P$ can be computed by an algebraic formula of size $2d+1$.
\end{proposition}
\begin{proof}
Follows from the fact that $P(x) = (\cdots((p_d x + p_{d-1})x + p_{d-2})\cdots)x + p_0$, which is a formula of size $2d+1$.
\end{proof}

The following observation shows that the classes of algebraic formulas/ABPs/circuits are ``robust'' under some very natural operations; we shall call such models as \emph{natural models}.
These properties are precisely the ones that we rely on in this paper.
Any \emph{natural model} would be sufficient for our purposes but it might be convenient for the reader to focus on just the standard models of formulas, ABPs and circuits.

\begin{definition}[Natural algebraic models]
An algebraic model $\mathcal{A}$ is called a \emph{natural model} if the class of polynomials computed by it satisfies the following properties.
  \begin{itemize}\itemsep 0pt
  \item Any polynomial of degree~$d$ with at most $s$ monomials has $\mathcal{A}$-size at most $s \cdot d$. In the specific setting when the polynomial is a univariate, its $\mathcal{A}$-size is $O(d)$.
  \item Partial substitution of variables by constants does not increase the $\mathcal{A}$-size of any polynomial.
  \item If each of $Q_1,\ldots, Q_k$ is computable in $\mathcal{A}$-size $s$, then $\sum Q_i$ is computable in $\mathcal{A}$-size at most $sk$.
  \item Suppose $P(x_1,\ldots, x_n)$ is computable in $\mathcal{A}$-size $s_1$ and say $Q_1,\ldots, Q_n$ are polynomials each of which can be computed in $\mathcal{A}$-size $s_2$.
  Then, $P(Q_1,\ldots, Q_n)$ can be computed in $\mathcal{A}$-size at most $s_1\cdot s_2$. \qedhere
  \end{itemize}
\end{definition}

\begin{observation}
Algebraic circuits, branching programs (ABPs), and formulas are natural algebraic models.
\end{observation}

\subsection{Combinatorial designs}

\begin{definition}[Combinatorial designs
(Nisan~\cite{nisan91})]   
\label{defn:NWdesign}
  A family of sets $\set{S_1,\ldots, S_m}$ is said to be an $(\ell,k,r)$ design if
  \begin{itemize}\itemsep 0pt
  \item $S_i \subseteq [\ell]$,
  \item $|S_i| =k$,
  \item $\abs{S_i \cap S_j} < r$ for any $i \neq j$. \qedhere
  \end{itemize}
\end{definition}

\noindent
The following is a standard construction of such designs based on the \emph{Reed-Solomon} code, similar to the construction described in \cite[Lemma 2.5]{NW94}.

\begin{lemma}[Construction of designs]\label{lem:NW-construction}
  Let $c \geq 2$ be any positive integer. There is an algorithm that, given parameters $\ell,k,r$ satisfying $\ell = k^c$ and $r \leq k$ with $k$ being a power of $2$, outputs an $(\ell,k,r)$ design $\set{S_1,\ldots, S_m}$ for $m \leq  k^{(c-1)r}$ in time $\poly(m)$.
\end{lemma}
\begin{proof}
  Since $k$ is a power of $2$, we can identify $[k]$ with the field $\F_k$ of $k$-elements and  $[\ell]$ with $\F_k \times \F_{k^{c-1}}$. For each univariate polynomial $p(x) \in \F_{k^{c-1}}[x]$ of degree less than $r$, define the set $S_p$ as
  \[
    S_p = \setdef{(i,p(i))}{i\in \F_k}.
  \]
  Since there are $k^{(c-1)r}$ such polynomials we get $k^{(c-1)r}$ subsets of $\F_k \times \F_{k^{c-1}}$ of size~$k$ each. Furthermore, since any two distinct univariate polynomials cannot agree at $r$ or more places, it follows that $\abs{S_p \cap S_q} < r$ for $p \neq q$.
\end{proof}

\subsection{Hardness-randomness connections}

\begin{observation}\label{obs:HS-for-multiples}
Let $H$ be a hitting set for the class $\mathcal{C}(n,d,s)$ of $n$-variate polynomials of degree at most $d$ that are computable by formulas of size~$s$. Then, for any nonzero polynomial $Q(x_1,\ldots, x_n)$ such that $\deg(Q) \leq d$ and $Q(\veca) = 0$ for all $\veca \in H$, we have that $Q$ cannot be computed by formulas of size~$s$.
\end{observation}
\begin{proof}
  If $Q$ was indeed computable by formulas of size at most $s$, then $Q$ is a member of $\mathcal{C}(n,d,s)$ for which $H$ is a hitting set. This would violate the assumption that $H$ was a hitting set for this class as $Q$ is a nonzero polynomial in the class that vanishes on all of $H$.
\end{proof}

From this observation, it is easy to see that explicit hitting sets can be used to construct lower bounds.

\begin{lemma}[Hitting sets to hardness \cite{HS80,A05a}] \label{lem:HS-to-hardness}
  Let $H$ be an explicit hitting set for $\mathcal{C}(n,d,s)$.
Then, for any $k\leq n$ such that $k|H|^{\frac{1}{k}} \leq d$, there is a polynomial $Q(z_1,\ldots, z_k)$ of individual degree smaller than $|H|^{\frac{1}{k}}$ that is computable in time $\poly(|H|)$ that requires formulas of size~$s$ to compute it. Furthermore, given the set $H$, there is an algorithm to output a formula of size $|H| \cdot d$ for $Q$ in time $\poly(|H|)$.
\end{lemma}
\begin{proof}
  This is achieved by finding a nonzero $k$-variate polynomial, for $k \leq n$, of individual degree  $d' < |H|^{\frac{1}{k}}$, that vanishes on the hitting set $H$; this can be done by interpreting it as a homogeneous linear system with $(d'+1)^k$ ``variables'' and at most $|H|$ ``constraints''.
Such a $Q_k$ can then be found via interpolation by solving a system of linear equations in time $\poly(|H|)$.
The degree of $Q_k$ is at most $k \cdot |H|^{\frac{1}{k}} \leq d$ from the hypothesis and the hardness of $Q_k$ follows from \expref{Observation}{obs:HS-for-multiples}.
\end{proof}

\begin{remark}[Bit complexity of $Q_k$]
\label{rmk:hard-poly-bit-size}
Note that we can obtain a hard polynomial $Q_k$ such that its coefficients
have bit-lengths that are at most polynomially large in terms of the
(total) bit-lengths of the points in the given hitting set. %
Moreover, one can also ensure that such a computation only handles numbers
of polynomially larger
bit-lengths\footnote{Edmonds has shown that Gaussian elimination runs in %
polynomial time in terms of the total bit-length of the input~\cite{edmonds67},
see \cite[Theorem 3.3]{schrijver}.  The same effect is achieved by
Bareiss's algorithm~\cite{B68}, a modification of Gaussian elimination.}, %
thereby making \expref{Lemma}{lem:HS-to-hardness} constructive in 
an \emph{explicit} sense.
\end{remark}

It is also known that we can get non-trivial hitting sets from suitable hardness assumptions.
For a fixed $(\ell,k,r)$ design $\set{S_1,\ldots, S_m}$ and a polynomial $Q(z_1,\ldots, z_k) \in \F[\vecz]$ we shall use the notation $\NWify{Q}{\ell,k,r}$ to denote the vector of polynomials
\[
\NWify{Q}{\ell,k,r} :=  \inparen{Q(\vecy\mid_{S_1}),Q(\vecy\mid_{S_2}),\ldots, Q(\vecy\mid_{S_m})} \in \inparen{\F[y_1,\ldots, y_\ell]}^m.
\]

Kabanets and Impagliazzo~\cite{KI04} showed that, if $Q(\vecz_{[k]})$ is hard enough, then $P(\NWify{Q}{\ell,k,r})$ is nonzero if and only if $P(\vecx_{[m]})$ is nonzero.
However, their proof crucially relies on a result of Kaltofen~\cite{K89} (or even a non-algorithmic version due to B\"{u}rgisser~\cite{B00}) about the complexity of factors of polynomials.
Hence, this connection is not directly applicable while working with other subclasses of circuits such as algebraic formulas as we do not know if they are closed under factorization.
The following lemma can be used in such settings and this paper makes heavy use of this.

\begin{lemma}[Hardness to randomness without factor complexity]\label{lem:KI-without-Kaltofen}
  Let $Q(z_1,\ldots,z_k)$ be an arbitrary polynomial of individual degree smaller than $d$.
Suppose there is an $(\ell,k,r)$ design $\set{S_1,\ldots, S_m}$ and a nonzero polynomial $P(x_1,\ldots, x_m)$, of degree at most $D$, that is computable by a formula of size at most $s$ such that $P(\NWify{Q}{\ell,k,r}) \equiv 0$.
Then there is a polynomial $\tilde{P}(z_1,\ldots, z_k)$, whose degree is at most $k\cdot d\cdot D$ that is divisible by $Q$ and computable by formulas of size at most $s \cdot (r-1) \cdot d^{r} \cdot (D+1)$.

Moreover, if $r = 2$, then this upper bound can be improved to $4\cdot s \cdot d \cdot (D+1)$
\end{lemma}

If the polynomial $Q(z_1,\ldots, z_k)$ in the above lemma was chosen such that $Q$ vanished on some hitting set $H$ for the class of size-$s'$, $n$-variate, degree $d'$ polynomials where $s' \geq s \cdot (r-1) \cdot d^r \cdot (D+1)$, then so does $\tilde{P}$ since $Q$ divides it.
If it happens that $\deg(\tilde{P})\leq d'$, then \expref{Observation}{obs:HS-for-multiples} immediately yields that $\tilde{P}$ cannot be computed by formulas of size~$s'$, contradicting the conclusion of the above lemma.
Hence, in such instances, we would have that $P(\NWify{Q}{\ell,k,r}) \not\equiv 0$, without appealing to any result about closure of the particular model under factorization.

\begin{proof}%

Borrowing the ideas from Kabanets and Impagliazzo \cite{KI04}, we look at the $m$-variate substitution $\inparen{x_1,\ldots,x_m} \mapsto \NWify{Q}{\ell,k,r}$ as a sequence of $m$ univariate substitutions.
We now introduce some notation to facilitate this analysis.

Given the $(\ell,k,r)$ design $\set{S_1,\ldots,S_m}$, let $\vecy_{i} = \vecy\mid_{S_i}$, for each $i \in [m]$.
The tuple $\NWify{Q}{\ell,k,r}$ can therefore be written as $(Q(\vecy_1),Q(\vecy_2),\ldots,Q(\vecy_m)) \in \inparen{\F[y_1,\ldots,y_{\ell}]}^{m}$.
For each $0 \leq i \leq m$, let 
\[
P_i = P(Q(\vecy_1),Q(\vecy_2),\ldots,Q(\vecy_i),x_{i+1},\ldots,x_m) \, ,
\]
which is $P$ after substituting for the variables $x_1,\ldots,x_i$.
Since $P_0 = P$ is a nonzero polynomial and $P_m = P(\NWify{Q}{\ell,k,r}) \equiv 0$, let $t$ be the unique integer with $1 \leq t \leq m$, for which $P_{t-1} \not\equiv 0$ and $P_{t} \equiv 0$.

Since $P_t(\vecy,x_{t},\ldots,x_m)$ is a nonzero polynomial, there exist values that can be substituted to the variables besides $x_t$ and $\vecy_t$ such that it remains nonzero; let this polynomial be $P'_t(\vecy_t,x_t)$.
Also, for each $j \in [t-1]$, let $Q^{(t)}(\vecy_{j} \cap \vecy_t)$ be the polynomial obtained from $Q(\vecy_j)$ after this substitution, which is a polynomial of individual degree less than $d$ on at most $(r-1)$ variables.
We can now make the following observations about $P'(\vecy_t,x_t)$:
  \begin{itemize}\itemsep 0pt
  \item Each $Q^{(t)}(\vecy_{j}\cap \vecy_t)$ has a formula of size at most $(d(r-1)) \cdot d^{r-1}$, and thus $P'(\vecy_t,x_t)$ has a formula of size at most $\inparen{s \cdot (r-1) \cdot d^{r}}$,
    \item $\deg(P') \leq D \cdot \deg(Q) \leq D \cdot (kd)$, and $\deg_{x_t}(P') \leq D$,
    \item $P'(\vecy_t,Q(\vecy_t)) \equiv 0$.
  \end{itemize}

  The last observation implies that the polynomial $(x_t - Q(\vecy_t))$ divides $P'$. Therefore we can write $P' = (x_t - Q(\vecy_t)) \cdot R$, for some polynomial $R$. Consider $P'$ and $R$ as univariates in $x_t$ with coefficients as polynomials in $\vecy_t$:
  \[
    P' = \sum_{i=0}^D P'_i \cdot x_t^i\quad,\quad R  = \sum_{i=0}^{D-1} R_i \cdot x_t^i.
  \]
  If $a$ is the smallest index such that $P'_a\neq 0$, then $P'_a = - R_a \cdot Q(\vecy_t)$ and hence $Q(\vecy_t)$ divides $P'_a$.
  Any coefficient $P'_i$ can be obtained from $P'$ using interpolation from $(D+1)$ evaluations of $x_t$. Hence, $\tilde{P} = P'_a$ can be computed in size $\inparen{s \cdot (r-1) \cdot d^{r} \cdot (D+1)}$.

 For the case of $r=2$, observe that the polynomial $Q^{(t)}(\vecy_{j}\cap \vecy_t)$ is a univariate of degree at most $d$. Thus, by~\expref{Proposition}{prop:horner}, $Q^{(t)}(\vecy_{j}\cap \vecy_t)$ can be computed by a formula of size $2d + 1 \leq 4d$. So, we get an upper bound of $\inparen{4\cdot s\cdot d}$ on the formula complexity of $P'(\vecy_t,x_t)$ (instead of $O(sd^2)$ that we would get by invoking the general bound for $r = 2$) and after interpolation as above, we get a bound of $4\cdot s\cdot d \cdot (D+1)$ on the formula complexity of $P_a'$ as defined above.
\end{proof}

\section{Bootstrapping Hitting Sets}

The following are the main bootstrapping lemmas to yield our main result.
These lemmas follow the same template as in the proof of Agrawal, Ghosh and Saxena~\cite{AGS19} but with some simple but important new ideas that avoid any requirement on bounds on factor complexity, and also permitting a result starting from a barely non-trivial hitting set.

\begin{restatable}[Barely non-trivial to moderately non-trivial hitting sets]{lemma}{stepzerolemma}
  \label{lem:step-zero}
  Let $\epsilon >0$ and $n \geq 2$ be constants. Suppose that for all large enough $s$ there is an explicit hitting set of size $s^{n - \epsilon}$, for all degree-$s$, size-$s$ algebraic formulas over $n$ variables.

  Then for a large enough $m \geq n^8$, and for all $s \geq m$, there is an explicit hitting set of size $s^{\frac{m}{50}}$ for all degree-$s$, size-$s$ algebraic formulas over $m$ variables.
\end{restatable}

\begin{restatable}[Bootstrapping moderately non-trivial hitting sets]{lemma}{bootstraplemma}
  \label{lem:main-bootstrap}
  Let $n_0$ be large enough, and $n$ be any power of two that is larger than $n_0$. Suppose for all $s \geq n$ there are explicit hitting sets of size $s^{g(n)}$ for $\mathcal{C}(n,s,s)$, the class of $n$-variate degree-$s$ polynomials computed by size-$s$ formulas.
  \begin{enumerate}
    \item\label{lbl:main-bootstrap-a} Suppose $g(n) \leq \frac{n}{50}$, then for $m = n^{10}$ and all $s \geq m$, there are explicit hitting sets of size $s^{h(m)}$ for $\mathcal{C}(m,s,s)$ where $h(m) \leq \pfrac{1}{10} \cdot m^{\frac{1}{4}}$.
    \item\label{lbl:main-bootstrap-b} Suppose $g(n) \leq \pfrac{1}{10} \cdot n^{\frac{1}{4}}$, then for $m = 2^{n^{\frac{1}{4}}}$ and all $s \geq m$, there are explicit hitting sets of size $s^{h(m)}$ for $\mathcal{C}(m,s,s)$ where $h(m) = 20\cdot \inparen{g(\log^{4}m)}^{2}$.

      Furthermore, $h(m)$ also satisfies $h(m) \leq \pfrac{1}{10} \cdot m^{\frac{1}{4}}$.
  \end{enumerate}
\end{restatable}

\noindent
We will defer the proofs of these lemmas to the end of this section and complete the proof of \expref{Theorem}{thm:mainTheorem-2}.

\maintheorem*

\begin{proof}
  Notice that \expref{Lemma}{lem:step-zero} and \expref{Lemma}{lem:main-bootstrap} are structured so that the conclusion of \expref{Lemma}{lem:step-zero} is precisely the hypothesis of  \expref{Lemma}{lem:main-bootstrap}(\ref{lbl:main-bootstrap-a}), the conclusion of  \expref{Lemma}{lem:main-bootstrap}(\ref{lbl:main-bootstrap-a}) is precisely the hypothesis of \expref{Lemma}{lem:main-bootstrap}(\ref{lbl:main-bootstrap-b}), and \expref{Lemma}{lem:main-bootstrap}(\ref{lbl:main-bootstrap-b}) admits repeated applications as its conclusion also matches the requirements in the hypothesis.
Thus, we can use one application of \expref{Lemma}{lem:step-zero} followed by one application of \expref{Lemma}{lem:main-bootstrap}(\ref{lbl:main-bootstrap-a}) and repeated applications of \expref{Lemma}{lem:main-bootstrap}(\ref{lbl:main-bootstrap-b}) to  get hitting sets for polynomials depending on larger sets of variables, until we can get a hitting set for the class $\mathcal{C}(s,s,s)$.

Let $n_0$ be large enough so as to satisfy the hypothesis of \expref{Lemma}{lem:step-zero}, and the two parts of \expref{Lemma}{lem:main-bootstrap}.
We start with an explicit hitting set of size $s^{n_0 - \epsilon}$ for $\mathcal{C}(n_0,s,s)$ and one application of \expref{Lemma}{lem:step-zero} gives an explicit hitting set of size $s^{\frac{n_1}{50}}$ for $\mathcal{C}(n_1,s,s)$ for $n_1 \geq n_0^8$ and all $s \geq n_1$.
Using \expref{Lemma}{lem:main-bootstrap}(\ref{lbl:main-bootstrap-a}) we obtain an explicit hitting set of size $s^{(1/10)\cdot m_0^{\frac{1}{4}}}$ for the class $\mathcal{C}(m_0,s,s)$ for all $s \geq m_0 = n_1^{10}$. We are now in a position to apply \expref{Lemma}{lem:main-bootstrap}(\ref{lbl:main-bootstrap-b}) repeatedly. We now set up some basic notation to facilitate this analysis.

Suppose after $i$ applications of \expref{Lemma}{lem:main-bootstrap}(\ref{lbl:main-bootstrap-b}) we have an explicit hitting set for the class $\mathcal{C}(m_i,s,s)$ of size $s^{t_i}$.
We wish to track the evolution of $m_i$ and $t_i$. Recall that $m_{i} = 2^{m_{i-1}^{\frac{1}{4}}}$ after one application of \expref{Lemma}{lem:main-bootstrap}(\ref{lbl:main-bootstrap-b}).

Let $\{b_i\}_{i}$ be such that $b_0 = \log{m_0}$ and, for every $i > 0$, let $b_i = 2^{(b_{i-1}/4)}$ so that $b_i = \log m_i$.
Similarly to keep track of the complexity of the hitting set, if $s^{t_i}$ is the size of the hitting set for $\mathcal{C}(m_i,s,s)$, then by \expref{Lemma}{lem:main-bootstrap}(\ref{lbl:main-bootstrap-b})  we have $t_0 = \pfrac{1}{10} m_0^{\frac{1}{4}}$ and $t_i = 20\cdot t_{i-1}^2$ for all $i> 0$.

  \noindent
  The following facts are easy to verify.
  \begin{itemize}
    \item $m_i \geq s$ or $b_i \geq \log s$ for $i = O(\log^{*}s)$,
    \item for all $j$, we have $t_j = 20^{(2^j - 1)} \cdot t_0^{2^j} = \exp(\exp(O(j)))$.
    \item the exponent of $s$ in the complexity of the final hitting set is $t_{O(\log^{*}s)} = \exp(\exp(O(\log^\ast s)))$.
  \end{itemize}
  Therefore we have an explicit hitting set of size $s^{\exp(\exp(O(\log^{*}s)))}$ for $\mathcal{C}(s,s,s)$. An explicit algorithm describing the hitting set generator is presented in \expref{Subsection}{subsec:algo}.
\end{proof}

\subsection{Proofs of the bootstrapping lemmas}

Here we prove the two main lemmas used in the proof of \expref{Theorem}{thm:mainTheorem-2}. We restate the lemmas here for convenience. The proofs follow a very similar template but with different settings of parameters and minor adjustments.

\stepzerolemma*
\begin{proof}
Let $a = \max(n,\frac{250}{\epsilon})$.
We begin by fixing the design parameters, $k = n$, $\ell = a \cdot k^4 = a \cdot n^4$ and $r=2$.

  \begin{description}
    \item[Constructing a suitably hard polynomial:]
    For $B = \frac{5k}{\epsilon}$, we construct a polynomial $Q_k(z_1,\ldots,z_k)$ that vanishes on the hitting set for all size-$s^{B}$ degree-$s^{B}$ formulas over $k$ variables, that has size $s^{B(k - \epsilon)}$ using \expref{Lemma}{lem:HS-to-hardness}. The polynomial $Q_k(\vecz)$ has the following properties.
    \begin{itemize}
      \item $Q_k$ has individual degree $d < s^{B(k - \epsilon)/k}$, and total degree $< k \cdot s^{B(k-\epsilon)/k}$.
      \item $Q_k$ is not computable by formulas of size $s^{B}$.
      \item $Q_k$ has a formula of size $\leq (kd) \cdot s^{B(k-\epsilon)}$.
    \end{itemize}

    \item[Building the NW design:]
       Using~\expref{Lemma}{lem:NW-construction}, we now construct an $(\ell,k,r)$ design $\set{S_1,\ldots,S_{m}}$ with  $m := \inparen{\frac{\ell}{k}}^r = \inparen{a k^{(4-1)}}^2 = a^2 k^6$.

    \item[Variable reduction:]
    Let $P(x_1,\ldots,x_{m})$ be a nonzero $m$-variate degree-$s$ polynomial computable by a formula of size~$s$, and let $P(\NWify{Q_k}{\ell,k,r}) \equiv 0$.
    Then, from the `moreover' part of \expref{Lemma}{lem:KI-without-Kaltofen} (since $r = 2$), we get  that there is a polynomial $\tilde{P}(z_1,\ldots,z_k)$ that vanishes on a hitting set for formulas of size $s^{B}$ and degree $s^{B}$, and is computable by a formula of size at most
    \begin{align*}
      \operatorname{size}(\tilde{P}) & \leq 4 \cdot s \cdot d \cdot (s+1)  \\
                                     &\leq 4 s(s+1) \cdot s^{B(k-\epsilon)/k}\\
                                              &\leq s^{\inparen{5 + \frac{B(k-\epsilon)}{k}}} =  s^{5 + \frac{5k}{\epsilon} - 5} = s^{B}.
    \end{align*}
    Moreover, note that the degree of $\tilde{P}(z_1,\ldots,z_k)$ is at most $(k \cdot d) \cdot s \leq s^{\inparen{2 + \frac{B(k-\epsilon)}{k}}} < s^B$. Since $\tilde{P}$ vanishes on the hitting set for formulas of size $s^{B}$ and degree $s^{B}$, we get a contradiction due to \expref{Observation}{obs:HS-for-multiples}. Therefore, it must be the case that $P(\NWify{Q_k}{\ell,k,r})$ is nonzero.%

    \item[Construction of the hitting set:]
      Therefore, starting with a nonzero formula of degree~$s$, size~$s$, over $m$ variables, we obtain a nonzero $\ell$-variate polynomial of degree at most $s \cdot (kd) \leq s^{B}$. At this point we can just use the trivial hitting set given by the
Polynomial Identity Lemma (\expref{Lemma}{lem:PIL})  %
which has size at most $s^{B \ell}$.

    Therefore, what remains to show is that our choice of parameters ensures that $B \ell < \frac{m}{50}$. This is true, as $\frac{m}{50} = \frac{a^2 k^6}{50} = \frac{ak}{50} \cdot a k^5 = \pfrac{5k}{\epsilon} \cdot \ell \cdot k > B \ell$.%
\end{description}
\noindent
The construction runs in time that is polynomial in the size of the hitting set in the conclusion, and the \emph{bit-size} of the points in it. See \expref{Subsection}{subsec:algo} for a more elaborate discussion.
\end{proof}

\bootstraplemma*

\begin{proof}
  The proofs of both parts follow the same template as in the proof of \expref{Lemma}{lem:step-zero} but with different parameter settings. Hence, we will defer the choices of the parameters $\ell,k,r$ towards the end to avoid further repeating the proof. For now, let $\ell,k,r$ be parameters that satisfy $r \leq k$,  $\ell = k^2$ and $5r\cdot g(n)\leq k$.

  \begin{description}
  \item[Constructing a hard polynomial:]
    The first step is to construct a polynomial $Q_k(z_1,\ldots, z_k)$ that vanishes on the hitting set for the class $\mathcal{C}(n,s^{5},s^{5})$, where\footnote{that is, $Q_k$ is a $k$-variate polynomial that is just masquerading as an $n$-variate polynomial that does not depend on the last $n-k$ variables.} $k \leq n$. This can be done by using \expref{Lemma}{lem:HS-to-hardness}. The polynomial $Q_k(\vecz)$ will therefore have the following properties.
    \begin{itemize}\itemsep 0pt
    \item $Q_k$ has individual degree~$d$ smaller than $s^{5 g(n)/k}$, and degree at most $k \cdot s^{5 g(n)/k}$.
    \item Computing $Q_k$ requires formulas of size more than $s^{5}$.
    \item $Q_k$ has a formula of size at most $s^{10 g(n)}$.
    \end{itemize}

  \item[Building the NW design:]
    Using the parameters $\ell$, $k$, $r$, and the construction from \expref{Lemma}{lem:NW-construction}, we now construct an $(\ell,k,r)$ design $\set{S_1,\ldots, S_{m}}$ with $m \leq k^{r}$.

  \item[Variable reduction using $Q_k$:]
    Let $P(x_1,\ldots, x_m) \in \mathcal{C}(m,s,s)$ be a nonzero polynomial. Suppose for contradiction that $P(\NWify{Q_k}{\ell,k,r}) \equiv 0$, then \expref{Lemma}{lem:KI-without-Kaltofen} states that there is a nonzero polynomial $\tilde{P}(z_1,\ldots, z_k)$ of degree at most $s \cdot k\cdot d$ such that
    $Q_k$ divides $\tilde{P}$, and that $\tilde{P}$ can be computed by a formula of size at most
    \begin{align*}
      s \cdot (r-1) \cdot d^{r} \cdot (s+1) & \leq s^4  \cdot d^{r}\\
                                       & \leq  s^4 \cdot s^{5r\cdot g(n)/k}\\
                               & \leq s^5. & \text{(since $k,r$ satisfy $5r \cdot g(n) \leq k$)}
    \end{align*}
    Furthermore, the degree of $\tilde{P}$ is at most $s \cdot r\cdot s^{5g(n)/k} \leq s^5$. Hence, $\tilde{P}$ is a polynomial in $k \leq n$ variables, of degree at most $s^5$ that vanishes on the hitting set of $\mathcal{C}(n,s^5,s^5)$ since $Q_k$ divides $\tilde{P}$.
    But then, \expref{Observation}{obs:HS-for-multiples} states that $\tilde{P}$ must require formulas of size more than $s^5$, contradicting the above size bound. Hence, it must be the case that $P(\NWify{Q_k}{\ell,k,r}) \not\equiv 0$.

  \item[Hitting set for $\mathcal{C}(m,s,s)$:] At this point, we set the parameters $k$ and $r$ depending on how quickly $g(n)$ grows.
    \begin{description}
    \item[Part (1) $\inparen{g(n) \leq \frac{n}{50}}$:]

      In this case, we choose $k = n$ and $r = 10$ (so we satisfy $5r\cdot  g(n) \leq n = k$). From \expref{Lemma}{lem:NW-construction}, we have an explicit $(\ell,k,r)$ design $\set{S_1,\ldots, S_m}$ with $m = k^{r} = n^{10}$.

      For any  nonzero $P \in \mathcal{C}(m,s,s)$, we have that $P(\NWify{Q_k}{\ell,k,r})$ is a nonzero $\ell$-variate polynomial of degree at most $s \cdot k \cdot s^{5 g(n) / k} \leq s^3$. Hence, by just using the trivial hitting set via the
Polynomial Identity Lemma (\expref{Lemma}{lem:PIL})  %
we have an explicit hitting set of size $s^{3 \ell} \leq s^{3 m^{\frac{1}{5}}}$. Since $m \geq n_0$ and $n_0$ is large enough, we have that
      \[
        h(m) := 3 m^{\frac{1}{5}} \leq \pfrac{1}{10} \cdot m^{\frac{1}{4}}.
      \]
    \item[Part (2) $\inparen{g(n) \leq \pfrac{1}{10} n^{\frac{1}{4}}}$:] In this case, we choose $k = \sqrt{n}$ and $r = n^{\frac{1}{4}}$, so that  $5r\cdot g(n) \leq 10r \cdot g(n) \leq k$ and $\ell = n$.
    Using \expref{Lemma}{lem:NW-construction}, we now construct an explicit $(\ell,k,r)$ design $\set{S_1,\ldots, S_m}$ with $m = 2^{n^{\frac{1}{4}}} \leq k^{r}$.

      We have a formula computing the $n$-variate polynomial $P(\NWify{Q_k}{\ell,k,r})$ of size at most $s \cdot s^{10g(n)} \leq s^{20 g(n)} =: s'$. Using the hypothesis for hitting sets for $\mathcal{C}(n,s',s')$, we have an explicit hitting set for $\mathcal{C}(m,s,s)$ of size at most
      \[
        \inparen{s'}^{g(n)}  = s^{20 g(n)^2} = s^{h(m)},
      \]
      where $h(m) = 20 \inparen{g((\log m)^4)}^2$. Since $n_0$ is large enough, we have that
      \begin{align*}
        10 \cdot h(m) & \leq 20 \cdot 10 \cdot \inparen{g((\log m)^4)}^2\\
                     & \leq 2 \inparen{\log m}^2 & \text{(since $g(n) \leq \pfrac{1}{10} n^{\frac{1}{4}}$)}\\
                     & \leq m^{\frac{1}{4}}. & \text{(since $m \geq n_0$ and $n_0$ is large enough)}
      \end{align*}
    \end{description}
  \end{description}
    This completes the proof of both parts of the lemma.
\end{proof}

\section{Finer analysis of the main result}
\label{sec:finer-analysis}

This section addresses some of the subtle aspects about \expref{Theorem}{thm:mainTheorem-2} as follows.
First, we provide an algorithm outlining all the steps in obtaining the final hitting set from the initial one, which also helps us illustrate the explicitness of all the intermediate hitting sets (including the corresponding bit complexities).
We then discuss how our theorem statement changes if the hypothesis holds only when the number of variables is a growing function of the size/degree~$s$.

\subsection{Algorithm for generating the hitting set}
\label{subsec:algo}

We now give an algorithm to generate an explicit hitting set for $\mathcal{C}(s,s,s)$, for all large $s$, using the hypothesis of \expref{Theorem}{thm:mainTheorem-2}. Let $n_0$ be the initial threshold from the hypothesis and let $n_1$ be a constant that satisfies the ``large enough'' requirements of \expref{Lemma}{lem:main-bootstrap}.

\vskip -0.3in

\begin{align*}
  n_0 & \geq 2 & t_0 & = (n_0 - \epsilon)\\
  n_1 & \text{ is large enough} & t_1 &= \frac{n_1}{50}\\
  n_2 & = n_1^{10} & t_2 &= \inparen{\frac{1}{10}} n_2^{\frac{1}{4}}\\
  \text{For all $i \geq 3$,} \quad n_i & = 2^{n_{i-1}^{\frac{1}{4}}} & t_i &:= 20 t_{i-1}^2
\end{align*}
We are provided an algorithm $\textsc{Initial-Hitting-Set}(s)$ that outputs a hitting set for $\mathcal{C}(n_0,s,s)$ of size at most $s^{n_0 - \epsilon}$. \expref{Algorithm}{alg:final-algorithm} describes a function \textsc{Hitting-Set} which, given inputs $i$ and $s$, outputs a hitting set for $\mathcal{C}(n_i,s,s)$ of size at most $s^{t_i}$ in time $\poly(s^{t_i})$.

\begin{algorithm}
  \caption{\textsc{Hitting-Set}}
  \label{alg:final-algorithm}

  \DontPrintSemicolon
  \SetKwInOut{Input}{Input}\SetKwInOut{Output}{Output}
  \Input{~Parameter $i$ and a size $s$.}
  \Output{~A hitting set of size $s^{t_i}$ size for $\mathcal{C}(n_i,s,s)$.}

  \BlankLine

  \If{$i = 1$}{
    Let $A = \max\inparen{n_0,\frac{250}{\epsilon}}$, $B = \frac{3n_0}{\epsilon}$\;
    Let $H_{0}(s^B) := \textsc{Initial-Hitting-Set}(s^B)$ \tcp*[r]{size at most $s^{B(n_0 - \epsilon)}$}
    \BlankLine
    Compute a nonzero polynomial $Q$ in $k = n_{0}$ variables of individual degree smaller than $s^{B t_{0} / k}$ that vanishes on $H_{0}(s^B)$. \tcp*[r]{takes $\poly(s^{Bt_{0}})$ time}
    Compute an $(A \cdot n_{0}^4,n_{0},2)$-design $\set{S_1,\ldots, S_{n_1}}.$\;
    \BlankLine
    Let $S \subseteq \F$ be of size at least $s^B$.\;
    \Return{$\setdef{\inparen{\NWify{Q}{\ell,k,r}}(\veca)}{\veca \in S^{A n_0^5}}$}\tcp*[r]{size at most $s^{B A n_0^5} \leq  s^{\frac{n_1}{50}} = s^{t_1}$}
  }
  \ElseIf{$i = 2$}{
    $H_{1}(s^5) := \textsc{Hitting-Set}(1,s^5)$\tcp*[r]{size at most $s^{5t_{1}}$}
    \BlankLine
    Compute a nonzero polynomial $Q$ in $k = n_{1}$ variables of individual degree smaller than $s^{5 t_{1} / k}$ that vanishes on $H_{1}(s^5)$. \tcp*[r]{takes $\poly(s^{5t_{1}})$ time}
    Compute an $(n_{1}^2,n_{1},10)$-design $\set{S_1,\ldots, S_{n_2}}.$\;
    \BlankLine
    Let $S \subseteq \F$ be of size at least $s^3$.\;
    \Return{$\setdef{\inparen{\NWify{Q}{\ell,k,r}}(\veca)}{\veca \in S^{n_{1}^2}}$}\tcp*[r]{size at most $s^{3n_1^2} \leq s^{0.1 \cdot n_2^{\frac{1}{4}}} = s^{t_2}$}
  }
  \ElseIf{$i \geq 3$}{
    $H_{i-1}(s^5) := \textsc{Hitting-Set}(i-1,s^5)$\tcp*[r]{size at most $s^{5t_{i-1}}$}
    \BlankLine
    Compute a nonzero polynomial $Q$ in $k = \sqrt{n_{i-1}}$ variables of individual degree smaller than $s^{5 t_{i-1} / k}$ that vanishes on $H_{i-1}(s^5)$. \tcp*[r]{takes $\poly(s^{5t_{i-1}})$ time}
    Compute an $(n_{i-1},\sqrt{n_{i-1}},\sqrt[4]{n_{i-1}})$-design $\set{S_1,\ldots, S_{n_i}}$\;
    \BlankLine
    $H_{i-1}(s^{20t_{i-1}}) := \textsc{Hitting-Set}(i-1,s^{20t_{i-1}})$\tcp*[r]{size at most $s^{20t_{i-1}^2}$}
    \Return{$\setdef{\inparen{\NWify{Q}{\ell,k,r}}(\veca)}{\veca \in H_{i-1}(s^{20 t_{i-1}})}$}\tcp*[r]{size at most $s^{20t_{i-1}^2} = s^{t_i}$}
  }
\end{algorithm}

From the growth of $n_i$, it follows that $n_b \geq s$ for $b = O(\log^*s)$ and $t_i = 20^{2^i-1} t_0^{2^i}$. Unfolding the recursion for $\textsc{Hitting-Set}(j,s)$, for any $j$, the algorithm makes at most $2^j$ calls to $\textsc{Initial-Hitting-Set}(s')$ for various sizes $s'$ satisfying
\[
  s' \leq s^{B \cdot 20^{j-1} \prod_{i=1}^{j-1} t_i} \leq s^{B t_{j-1}^2} = s^{O(t_j)}.
\]
Thus for $\textsc{Hitting-Set}(b,s)$, the algorithm makes at most $2^b$ calls to $\textsc{Initial-Hitting-Set}(s')$, for sizes $s'$ that are at most $s^{\exp(\exp(O(\log^*s)))}$. The overall running time is polynomial time in the size of the final hitting set which is $s^{t_b} = s^{\exp(\exp(O(\log^*s)))}$.

\paragraph{Bit complexity of the hitting sets.}
\label{para:bit-size}

We will now discuss the bit complexity of the hitting sets that are generated during the bootstrapping procedure. We will analyze \expref{Algorithm}{alg:final-algorithm} and show that any hitting set $H$ for $n$-variate formulas that the algorithm outputs, will have a bit complexity that is at most $\abs{H}^{f(n)}$, for $f(n) = \exp(O(\log^{\ast}n))$.

Let us first consider the case when $i \geq 3$. Suppose each evaluation point in $S_0 := H_{i-1} (s^5)$ is at most $h^{a}$ bits long, where $h = \abs{S_0}$ and $a=f(n_{i-1})$.
Using \expref{Remark}{rmk:hard-poly-bit-size}, we get that each coefficient of $Q$ is at most $h^{c a}$ bits long, for some other constant $c$. The output of \expref{Algorithm}{alg:final-algorithm} for this case will be evaluations of $Q$ on $S_1 := H_{i-1} (s^{20t_{i-1}})$.
Since $\abs{S_1} = h^{4 t_{i-1}}$, the bit complexity of $S_1$ is at most $h^{4 a t_{i-1}}$. Now $Q$ is a degree $< s^5$ polynomial with $O(h)$ monomials.
As a result any evaluation of $Q$ on a point from $S_1$ will have at most $O(\log{h}) \cdot \inparen{h^{c a} + s^5 \cdot h^{4 a t_{i-1}}} \leq h^{2a \cdot (4t_{i-1})} = \abs{S_1}^{2a}$, as $t_{i-1}$ is large enough. Since $n_{i} = 2^{n_{i-1}^{1/4}}$, we have that $f(n) = \exp(O(\log^{\ast}n))$.
For the cases when $i=1$ or $i=2$, since we use the trivial hitting sets in place of $S_1$ in the above discussion, the same bounds will continue to hold.

We want to remark that although the bit complexity of the output of $\textsc{Hitting-Set}(i,s)$ is not polynomial in the size $s^{t_i}$, the exponent $t_i$ is always larger than $f(n_i)$. As a result the time taken to generate the final hitting set in the conclusion of \expref{Theorem}{thm:mainTheorem-2} can still be bounded by $s^{\exp(\exp(O(\log^{\ast}s)))}$, albeit with a slightly larger constant in $O(\log^{\ast}s)$.

\subsection{Bootstrapping with a slightly weaker hypothesis}
\label{sec:growing n}

We now discuss the analogue of \expref{Theorem}{thm:mainTheorem-2} when the number of variables $n$ in the hypothesis grows with the size parameter $s$.
In particular, the hypothesis now guarantees a hitting set that is better than the brute force hitting set of size $(s+1)^n $ only when $n$ grows faster than a certain function of $s$. For example, suppose that for all large $n$ and $s$, the class $\mathcal{C}(n,s,s)$ has hitting sets of size\footnote{Here $\ilog{i}$ denotes $i$ iterated applications of $\log$, \ie{} $\ilog{3}s = \log\log\log s$.} at most $s^{n^{\delta}(\ilog{3}s)^{\delta}} $ for some small constant $\delta$. 
Note that this hypothesis is not ``helpful'' when $n$ is a constant, or even $(\ilog{3}s)^{o(1)} $.
Nevertheless, we shall show that even this scenario allows for a very similar bootstrapping procedure, and leads to hitting sets of size $s^{\poly(\ilog{2}s)}$ for the class $\mathcal{C}(s,s,s)$.

We now state the analogue of the ``single-step lemma'' (\expref{Lemma}{lem:main-bootstrap}) that helps us in proving our main theorem of this section (\expref{Theorem}{thm:bootstrapping-growing-s}, which is proved at the end of the section). 

\begin{lemma}\label{lem:bootstrapping-growing-s}
  Let $g,t:\N\rightarrow \N$ be non-decreasing functions such that $g(k) \leq k^{1/4}$ and $t(x^{\log x}) \leq 2 t(x)$ for all $x$. Let $1\leq m\leq s$ such that 
  \begin{equation}\label{eqn:bootstrapping-growing-s-constraint}
    20 \cdot g\inparen{(\log m)^4} \cdot t(s) \leq \log m.
  \end{equation}

  \noindent

  Let $(k_1, s_1) = \inparen{(\log m)^4,  s^5}$ and $(k_2, s_2) = \inparen{(\log m)^4, s^{100 g((\log m)^2) t(s)}}$. Suppose $\mathcal{C}(k_1,s_1, s_1)$ and $\mathcal{C}(k_2, s_2, s_2)$ have explicit hitting sets of size at most $s_1^{g(k_1) t(s_1)}$ and $s_2^{g(k_2)t(s_2)}$, respectively. 
  
  Then, we have an explicit hitting set for $\mathcal{C}(m,s,s)$ of size at most $s^{g'(m) t'(s)}$, where
  \begin{align*}
    g'(m) & \leq 400 \cdot \inparen{g\inparen{(\log m)^4}}^2\\
    t'(s) & \leq t(s)^2.
  \end{align*}
\end{lemma}
\begin{proof}
  The proof is exactly along similar lines, with just a little more care to set parameters appropriately. 

  Fix any $m,s$ that satisfy \expeqref{equation}{eqn:bootstrapping-growing-s-constraint}.
  Set $r = \log m$, $k = (\log m)^2$ and $\ell = (\log m)^4$. With these parameters, \expeqref{equation}{eqn:bootstrapping-growing-s-constraint} simplifies to $20 \cdot g(\ell) \cdot t(s) \leq r$. 

Let $S_1,\ldots, S_m$ be an $(\ell,k,r)$ design, as guaranteed by \expref{Lemma}{lem:NW-construction}.
Let $s_1 = s^{5}$.
By the hypothesis for $(k_1=\ell,s_1)$ , there is an explicit hitting set $H$ of size at most $s_1^{g(\ell)t(s_1)} \leq s^{5 g(\ell) \cdot 2\cdot t(s)} \leq s^{r}$  for $\mathcal{C}(k,s_1,s_1)$.
Let $Q(z_1,\ldots, z_k)$ be the explicit polynomial of individual degree at most $s^{10g(\ell)t(s)/k} \leq s^{r/k}$, as guaranteed by \expref{Lemma}{lem:HS-to-hardness}, that vanishes on the set $H$.
Therefore, $Q_k$ satisfies the following properties:
  \begin{itemize}\itemsep 0pt
  \item $Q_k$ has individual degree less than $s^{r/k}$ and total degree at most $k \cdot s^{r/k}$. 
  \item $Q_k$ can be computed by formulas of size at most  $s^{20g(\ell)t(s)} \leq s^{r}$.

  \end{itemize}
  
  Now, suppose $P(x_1,\ldots, x_m)$ is any nonzero polynomial in $\mathcal{C}(m,s,s)$ but $P(\NWify{Q_k}{\ell,k,r}) \equiv 0$, then \expref{Lemma}{lem:KI-without-Kaltofen} states that there is a nonzero polynomial $\tilde{P}(z_1,\ldots, z_k)$ such that
  \begin{itemize}
  \item $\tilde{P}$ has degree at most $s^3$, and is a multiple of $Q_k$
  \item $\tilde{P}$ is computed by formulas of size at most $s^3 \cdot s^{r^2/k} = s^4$
  \end{itemize}
  However, \expref{Observation}{obs:HS-for-multiples} asserts that $\tilde{P}$ must require formulas of size at least $s^{5}$ and hence it must be the case that $P(\NWify{Q_k}{\ell,k,r}) \not\equiv 0$.

  \medskip
  
  Let $P'(z_1,\ldots, z_\ell) = P(\NWify{Q_k}{\ell,k,r})$.
This is an $\ell$-variate nonzero polynomial of degree at most $s \cdot \deg(Q_k) \leq s^3$ and is computable by formulas of size $s_2 = s \cdot s^{20g(\ell)t(s)} \leq s^{20g(\ell)t(s)}$.
Since $20 g(\ell) t(s) \leq \log m \leq \log s$, we have that $t(s_2) \leq 2 t(s)$.
Using the hypothesis for $(k_2=\ell, s_2)$ an explicit hitting set for $\mathcal{C}(\ell, s_2, s_2)$ of size at most $s_2^{g(\ell)t(s_2)}$ which can be bounded above as follows:
  \begin{align*}
    s_2^{g(\ell) t(s_2)}   & \leq s_2^{2\cdot g(\ell) t(s)}\\
                       & \leq s^{20 \cdot g(\ell) t(s) \cdot 2\cdot g(\ell) \cdot t(s)} \\
                       & \leq s^{400 \cdot g(\ell)^2 \cdot t(s)^2} \leq s^{g'(m) t'(s)}. 
  \end{align*}
  Therefore, we have an explicit hitting set for $\mathcal{C}(m,s,s)$ of size at most $s^{g'(m) t'(s)}$.
\end{proof}

Intuitively, the above lemma states that if $g(k) t(s)$ is substantially smaller than $k$, then explicit hitting sets of size $s^{g(k)t(s)}$ for $\mathcal{C}(k,s,s)$ can be used to obtain an $s^{g'(m) t'(s)}$ hitting set for $\mathcal{C}(m,s,s)$ where $m$ is exponentially larger than $k$ (if $s$ is not too large in comparison to $m$).
Suppose $g'(m)t'(s)$ is also substantially smaller than $m$ (which would roughly translate to $g(k) t(s)$ being \emph{much} smaller than $k$; say something like $\log\log k$), then we may be in a position to use the lemma again to get an even better hitting set.

The following theorem sets up the parameters suitably when we are in a position to use the above lemma multiple times. 

\begin{theorem}\label{thm:bootstrapping-growing-s}
  Let $g_0, t_0: \N \rightarrow \N$ be non-decreasing functions with $g_0(m) \leq m^{1/4}$ and $t_0(s)\leq \log s$.
Let $g_1,g_2,\ldots:\N\rightarrow \N$ and $t_1,t_2,\ldots:\N\rightarrow \N$ be defined as $g_{i+1}(m) = 400 \cdot \inparen{g_i\inparen{(\log m)^4}}^2$ and $t_{i+1}(s) = \inparen{t(s)}^2$.
Furthermore, let $L_0:\N\rightarrow \N$ as $L_1(x) = (\log x)^2$ and $L_{i+1}(x) = (L_{i}((\log x)^4))^2$ for all $i \geq 1$.

  \medskip

  Suppose for every $m,s \geq 1$ we have an explicit hitting set for $\mathcal{C}(m,s,s)$ of size at most $s^{g_0(m) \cdot t_0(s)}$. Let $r_{m,s} \geq 1$ be the largest index such that
  \begin{itemize}\itemsep 0pt
  \item $g_r(x) \leq x^{1/4}$  and $t_r(x^{\log x}) \leq 2\cdot t_r(x)$ for all $x\geq 1$,
  \item $100 \cdot g_r(m) \cdot t_r(s) \leq L_r(m)$. 
  \end{itemize}
  Then for any $j \leq r_{m,s}$, we have an explicit hitting set for $\mathcal{C}(m,s,s)$ of size at most $s^{g_j(m)t_j(s)}$.

  In particular, $\mathcal{C}(s,s,s)$ has an explicit hitting set of size $s^{L_{r}(s)}$ for $r = r_{s,s}$. 
\end{theorem}

Although the above theorem is stated in a general form, it would be instructive to just consider \emph{nice} functions of the form $g_0(s) = s^{0.1}$. In this case, it can be seen that
\begin{align*}
  g_i(s) & = 2^{2^{O(i)}} \cdot \poly(\ilog{i} s),\\
  t_i(s) & = (\ilog{3} s)^{2^i}, \text{ and}\\
  L_i(s) & = 2^{2^{O(i)}} \cdot \poly(\ilog{i} s),
\end{align*}
where $g_1(s) \ll L_1(s)$ but, certainly for $i = \log^*s$ we have $g_i(s) > L_i(s)$. If $t_0$ is the constant function, then we are essentially in the same regime as in the previous sections, so we get $r_{s,s} = \Theta(\log^*s)$, and we recover \expref{Theorem}{thm:mainTheorem-2}. However, if suppose $t_0(s) = \ilog{3}s$ (where $\ilog{i}s$ is the iterated logarithm function, \ie{} $\ilog{3}s = \log\log\log s$), then $t_3(s) = (\ilog{3}s)^{2^3} > L_3(s)$ and hence $r_{s,s} = 2$. Nevertheless, the above theorem yields an explicit hitting set of size $\mathcal{C}(s,s,s)$ of size at most $s^{\poly\ilog{2} s}$ which is still substantially better than $s^{g_0(s)t_0(s)} > s^{s^{0.1}}$.

\begin{proof}
  We will prove the theorem by induction on $j$. For any $m,s$, the base case of $j=0$ trivially satisfied by the  hypothesis of the theorem and hence there is nothing to prove.

  \medskip

  \noindent
  {\bf Inductive step ($j\rightarrow j+1$):} Suppose $g_{j+1}(x) \leq x^{1/4}$ and $t_{j+1}(x^{\log x}) \leq 2\cdot t_{j+1}(x)$ for all $x\geq 1$.
For each $i \leq j+1$, define $\mathcal{D}_i = \setdef{(m,s)}{m\leq s \;,\; 100 \cdot g_i(m) t_i(s) \leq L_i(m)}$.
By the induction hypothesis, for every $(k,s') \in \mathcal{D}_j$ we have an explicit hitting set for $\mathcal{C}(k,s',s')$ of size at most $s'^{g_j(k) t_j(s')}$.

  \medskip

  Fix any $(m,s) \in \mathcal{D}_{j+1}$. To construct the explicit hitting set for $\mathcal{C}(m,s,s)$, we will use \expref{Lemma}{lem:bootstrapping-growing-s}. In order to do that, we need to assert the following  claims:
  \begin{enumerate}[(a)]\itemsep 0pt
  \item\label{item:bs-growing-s-claim-1} $(m,s)$ satisfy \expeqref{equation}{eqn:bootstrapping-growing-s-constraint} (for $g = g_j$ and $t = t_j$),
  \item\label{item:bs-growing-s-claim-2} $\mathcal{C}(\ell,s_1,s_1)$, for $\ell = (\log m)^4$ and $s = s^5$, has an explicit hitting set of size at most $s_1^{g_{j}(\ell) t_j(s)}$, and
  \item\label{item:bs-growing-s-claim-3} $\mathcal{C}(\ell,s_2,s_2)$, for $\ell = (\log m)^4$ and $s_2 = s^{20 g_j(\ell) t_j(s)}$, has an explicit hitting set of size at most $s_2^{g_j(\ell) t_j(s_2)}$. 
  \end{enumerate}

  \medskip
  
  \noindent
  \textit{Pf of \expeqref{Claim}{item:bs-growing-s-claim-1}:}  $(m,s) \in \mathcal{D}_{j+1}$ implies that
  \begin{align*}
    20 \cdot g_j((\log m)^4) \cdot t_j(s) & = 20 \cdot \inparen{\frac{g_{j+1}(m)}{400}}^{1/2} \cdot \inparen{t_{j+1}(m)}^{1/2}\\
                                           & = \inparen{g_{j+1}(m) t_{j+1}(s)}^{1/2}\\
                                           & \leq \frac{1}{10} \cdot L_{j+1}(m)^{1/2} \leq \log m.
  \end{align*}

  \medskip
  
  \noindent
  \textit{Pf of \expeqref{Claim}{item:bs-growing-s-claim-2}:}
  If $\ell = (\log m)^4$ and $s_1 = s^5$,
  \begin{align*}
    100 \cdot g_j(\ell) t_j(s_1) & \leq 200 \cdot  g_j(\ell) t_j(s)\\
                       & = 10 \cdot \inparen{g_{j+1}(m)\cdot t_{j+1}(s)}^{1/2}\\
                       & \leq \frac{10}{10} \cdot L_{j+1}(m)^{1/2}\\
                       & = L_j(\ell).
  \end{align*}
  Hence $(\ell, s_1) \in \mathcal{D}_j$ and by the inductive hypothesis, we have an explicit hitting set for $\mathcal{C}(\ell, s_1, s_1)$ of size at most $s_1^{g_j(\ell)t_j(s_1)}$.

  \medskip
  
  \noindent
  \textit{Pf of \expeqref{Claim}{item:bs-growing-s-claim-3}:} For $\ell = (\log m)^4$ and $s_2 = s^{20 g_j(\ell) t_j(s)}$, note that \expeqref{Claim}{item:bs-growing-s-claim-1} shows that $20 g_j(\ell) t_j(s) \leq \log m \leq \log s$ and hence $t_j(s_2) \leq 2 t_j(s)$.
Therefore,
  \begin{align*}
    100 \cdot g_j(\ell) t_j(s_2) & \leq 200 \cdot g_j(\ell) t_j(s) \\
                                & = 10 \cdot \inparen{g_{j+1}(m) \cdot t_{j+1}(s)}^{1/2}\\
                                & = \inparen{L_{j+1}(m)}^{1/2}  \leq L_j(\ell).
  \end{align*}
  Hence $(\ell, s_2) \in \mathcal{D}_{j}$ and, by the induction hypothesis, we have an explicit hitting set for $\mathcal{C}(\ell, s_2, s_2)$ of size at most $s_2^{g_j(\ell) t_j(s_2)}$.

  Using \expref{Lemma}{lem:bootstrapping-growing-s}, we have that $\mathcal{C}(m,s,s)$ has an explicit hitting set of size $s^{g_{j+1}(m) t_{j+1}(s)}$.
\end{proof}

\section{Open problems}
We conclude with some open questions.
\begin{itemize}

  \item A natural question in the spirit of the results in this paper, and those in \cite{AGS19} seems to be the following: Can we hope to bootstrap lower bounds?
  In particular, can we hope to start from a mildly non-trivial lower bound for general algebraic circuits (\eg,
  superlinear or just superpolynomial), and hope to amplify it to get a stronger lower bound (superpolynomial or truly exponential, respectively).
  In the context of non-commutative algebraic circuits, Carmosino, Impagliazzo, Lovett and Mihajlin~\cite{CILM18} recently showed such results, but no such result appears to be known for commutative algebraic circuits.

  \item It would also be interesting to understand if it is possible to bootstrap white box PIT algorithms. The proofs in this paper strongly rely on the fact that we have a non-trivial blackbox derandomization of PIT.

  \item Kabanets and Impagliazzo~\cite{KI04} show that given a polynomial family which requires %
  exponential-size   %
  circuits, there is a blackbox PIT algorithm which runs in quasipolynomial time ($\exp(\poly(\log n))$) on circuits of size $\poly(n)$ and degree $\poly(n)$, where $n$ is the number of variables. Thus, even with the best hardness possible, the running time of the PIT algorithm obtained is still no better than  quasipolynomially bounded. The question is to improve this running time to get a better upper bound than that obtained in~\cite{KI04}. In particular, can we hope to get a deterministic polynomial-time PIT assuming that we have explicit polynomial families of exponential hardness. This seems to be closely related to the question about bootstrapping lower bounds.

\end{itemize}

Since the first version of this paper appeared, there has been some work related to the problems studied in this paper.  Specifically, Guo, Kumar, Saptharishi and Solomon~\cite{GKSS19} showed that  it is possible to bootstrap even explicit hitting sets of size $s^{n} - 1$ for $n$-variate 
circuits of size~$s$ and degree~$s$, 
for a constant $n$, to obtain polynomial-time PIT for $s$-variate circuits of size and degree~$s$. The proof in \cite{GKSS19} goes via the direct construction of a hitting set generator with constant seed length assuming a sufficiently strong hardness assumption, thereby making progress on the last open question state above as well. One caveat of the results in \cite{GKSS19} is that the underlying field needs to be of sufficiently large characteristic or characteristic zero. Proving analogous results over fields of small  characteristic continues to be a very interesting open problem.  

\paragraph{Acknowledgments:}

Ramprasad and Anamay would like to thank the organizers of the Workshop on Algebraic Complexity Theory (WACT 2018) where we first started addressing this problem. Ramprasad and Anamay would also like to thank Suhail Sherif and  Srikanth Srinivasan for making an observation which led to the strengthening of an older version of this paper.
Mrinal is thankful to Michael Forbes, Josh Grochow and Madhu Sudan for insightful discussions.

We also thank the anonymous reviewers at Theory of Computing for valuable feedback. The discussion in \expref{subsection}{sec:growing n} in particular was a result of a query from one of the reviewers on whether an analogue of \expref{Theorem}{thm:mainTheorem-2} continues to hold when the parameter $n$ in the hypothesis is growing with $s$.

Ramprasad and Anamay would also like to acknowledge the support from the Department of Atomic Energy, Government of India, under project number RTI4001.

\bibliographystyle{tocplain}
\bibliography{bibstrings,v019a012,bibtail}

\begin{tocauthors}
\begin{tocinfo}[mrinal]
Mrinal Kumar\\
Reader\\
School of Technology and Computer Science\\
Tata Institute of Fundamental Research, Mumbai\\
mrinal\tocat{}tifr\tocdot{}res\tocdot{}in\\
\url{https://mrinalkr.bitbucket.io/}
\end{tocinfo}
\begin{tocinfo}[ramprasad]
Ramprasad Saptharishi\\
Associate Professor\\
School of Technology and Computer Science, TIFR, Mumbai\\
ramprasad\tocat{}tifr\tocdot{}res\tocdot{}in \\
\url{https://www.tifr.res.in/~ramprasad/}
\end{tocinfo}
\begin{tocinfo}[anamay]
Anamay Tengse\\
Postdoctoral Student\\
Efi Arazi School of Computer Science, Reichman Universiry, Herzliya\\
anamay\tocdot{}tengse\tocat{}gmail\tocdot{}com \\
\url{https://anamay.bitbucket.io/}
\end{tocinfo}
\end{tocauthors}

\begin{tocaboutauthors}

\begin{tocabout}[mrinal]
\textsc{Mrinal Kumar} is a faculty member in the School of Technology and Computer Science at the Tata Institute of Fundamental Research (TIFR), Mumbai. Before moving to TIFR in 2022, he spent a few years as an assistant professor in the Department of Computer Science and Engineering at IIT Bombay. He is generally interested in questions in Computational Complexity, Algebraic Complexity and Coding Theory. He did his \phd\ at Rutgers, where he was advised by Swastik Kopparty and Shubhangi Saraf, and spent his postdoctoral years at a convex combination of Harvard, the Simons Institute and the University of Toronto before moving to Mumbai in 2019. Much before that, he grew up in the then small town of Deoghar in the state of Jharkhand in India, and for the first fifteen years of his life dreamed of growing up to play cricket for a living. His interest in math was kindled during the course of his interaction with Dr. Jugal Kishore Singh, his math teacher in high school. Dr. Singh convinced him (and half his teammates in the cricket team!) that math could be interesting, challenging and yet accessible. His unconventional lectures were full of delightful anecdotes, fantastic questions and quite remarkably, even his exams were enjoyable, non-scary and yet interesting. Perhaps most importantly, he encouraged his students to think independently, which many of his former students believe is the most important thing that they ever learnt in school.     
\end{tocabout}

\begin{tocabout}[ramprasad]
\textsc{Ramprasad Saptharishi} is a faculty member in the School of Technology and Computer Science at the Tata Institute of Fundamental Research, Mumbai, and is generally interested in most questions of an algebraic nature. He did his \phd\ at Chennai Mathematical Institute, advised by Prof. Manindra Agrawal, and spent a couple of years at 
Microsoft Research, India and Tel Aviv University for his postdoc. His other interests include board games, reading ridiculously long web-series and writing software code. However, writing reasonable biographies continues to remain outside his area of expertise. 
\end{tocabout}

\begin{tocabout}[anamay]
\textsc{Anamay Tengse} is currently a postdoc at %
Reichman University in Herzliya, and was a postdoc at the University of Haifa for a couple of years before that.
This work was done when he was a \phd\ student in the School of Technology and Computer Science at TIFR, Mumbai, advised by Ramprasad Saptharishi.
Prior to that, he obtained an M.\,Tech. from the Department of Computer Science and Engineering at IIT Bombay, where he first got interested in the theoretical aspects of computing.
His interest in research comes largely from his interactions with the faculty and his friends from Goa College of Engineering, where he did his undergrad.
He spends a dangerously large chunk of his time indulging in music, cricket, quizzes on sporcle, trying to understand poetry and linguistics, and justifying his lack of interest in football (soccer) to people who question his Goan origin.
\end{tocabout}

\end{tocaboutauthors}

\end{document}